\documentclass[11pt]{article} %PDFLATEX
\usepackage{amsfonts}
\usepackage{amsfonts}
\usepackage{amsfonts}
\usepackage{amsmath}
\usepackage{mathrsfs}
\usepackage{amsmath}
\usepackage{mathrsfs}
\usepackage{mathrsfs}
\usepackage{amsmath}
\usepackage{fancyhdr}
\usepackage{mathrsfs,morefloats}
\usepackage{amsfonts}
\usepackage{natbib}
\usepackage{amsmath}
\usepackage{amsthm}
\usepackage{amssymb}
\usepackage{color}
\usepackage{lscape}
\usepackage{rotating}
\usepackage{caption2}
\usepackage{subfigure}
\usepackage{graphicx}
\usepackage{bm}

\usepackage{epstopdf}
\usepackage{enumerate}
\graphicspath{ {C:/Users/jean/Desktop/NUS/} }

\usepackage{amssymb}

\newcommand{\Date}[1]{\def\@Date{#1}}
\def\today{\number\day~\ifcase\month\or
 January\or February\or March\or April\or May\or June\or
 July\or August\or September\or October\or November\or December\fi~\number\year}

\def\be{\begin{equation}}
\def\ee{\end{equation}}
\def\bea{\begin{eqnarray}}
\def\eea{\end{eqnarray}}
\def\bd{\begin{displaymath}}
\def\ed{\end{displaymath}}
\def\bda{\begin{eqnarray*}}
\def\eda{\end{eqnarray*}}
\def\bsm{\begin{small}}
\def\esm{\end{small}}

\def\t0{\theta_0}

\def\ha1{\hat \beta_1}

\def\bnt{\begin{enumerate}}
\def\ent{\end{enumerate}}

\def\bsc{\begin{scriptsize}}
\def\esc{\end{scriptsize}}

\newtheorem{theorem}{Theorem}

\newtheorem{lemma}{Lemma}

\theoremstyle{definition}
\newtheorem{condition}{Condition}

\newtheorem{example}{Example}
\newtheorem{remark}{Remark}

\makeatletter
\newcommand{\figcaption}{\def\@captype{figure}\caption}
\newcommand{\tabcaption}{\def\@captype{table}\caption}
\makeatother

\renewcommand{\theequation}{\thesection.\arabic{equation}}

\usepackage{hyperref}

\begin{document}

\title{ \bf Distribution Regression}

\author{Xin Chen  ~~ Xuejun Ma ~~Wang Zhou \\ Department of Statistics and Applied Probability, National University of Singapore\\ stacx@nus.edu.sg~~ stamax@nus.edu.sg~~ stazw@nus.edu.sg}

%\author{Xin Chen  %\thanks{ Department of Statistics and Applied Probability, National University of Singapore,  6 Science Drive 2, 117546, Singapore. stacx@nus.edu.sg}
%~~ Xuejun Ma %\thanks{ Department of Statistics and Applied Probability, National University of Singapore,  6 Science Drive 2, 117546, Singapore. stamax@nus.edu.sg}
%~~Wang Zhou \thanks{Department of Statistics and Applied Probability, National University of Singapore,  6 Science Drive 2, 117546, Singapore. stacx@nus.edu.sg, stamax@nus.edu.sg, stazw@nus.edu.sg}}

%\date{\large The University of Melbourne, Temple University, and University of Rochester}

\date{}
\maketitle

\begin{abstract}
Linear regression is a fundamental and popular statistical method. There are various kinds of linear regression, such as mean regression and quantile regression. In this paper, we propose a new one called distribution regression, which allows broad-spectrum of the error distribution in the linear regression. Our method uses nonparametric technique to estimate regression parameters. Our studies indicate that our method provides a better alternative than mean regression and quantile regression under many settings, particularly for asymmetrical heavy-tailed distribution or multimodal distribution of the error term.
Under some regular conditions, our estimator is $\sqrt n$-consistent and possesses the asymptotically normal distribution. The proof of the asymptotic normality of our estimator is very challenging because our nonparametric likelihood function cannot be transformed into sum of independent and identically distributed random variables. Furthermore, penalized likelihood estimator is proposed and enjoys the so-called oracle property with diverging number of parameters. Numerical studies also demonstrate the effectiveness and the flexibility of the proposed method.
\end{abstract}

\begin{quote}
\end{quote}

\begin{quote}
\noindent
{\sl MSC2010 subject classifications}: Primary 62J05; secondary 62J07
\end{quote}

\thispagestyle{empty}
\pagenumbering{gobble}

\newpage
\pagenumbering{arabic}

\setcounter{page}{1}

\section{Introduction}
Linear regression is an important tool to explore the relationship between the response and predictors in statistics, with much existing work including mean regression and quantile regression.

To study regression, we start with $\{(x_{i},Y_{i}),i=1,\dots,n\}$, a random sample from population $(x,Y)$, where $x$ is a p-dimensional predictor, $Y$ is a univariate response, and $(x_{i},Y_{i})$'s satisfy:
\begin{equation}\label{eq1}
  Y_{i}=\nu+x_{i}^\top\beta + \varepsilon_{i}, ~~~~i=1,\dots,n.
\end{equation}
Here $\beta=(b_1,\cdots,b_p)^\top$ is a p-dimensional vector of unknown parameters, $\nu$ is the intercept term, and the error terms $\varepsilon_{i}$'s  are independent and identically distributed (i.i.d.) with unknown density function $f$, and are assumed to be independent of $x_i$'s.

The commonly used method to estimate $\beta$  is to find $\widehat \beta$ which minimizes the objective function:
\begin{equation}\label{eq2}
  n^{-1}\sum_{i=1}^{n}\phi(Y_{i}-\nu-x_{i}^\top\beta).
\end{equation}
If $\phi(t)=t^2$, then $\widehat \beta$ is the mean regression or least square estimator. If  $\phi_{\tau}(t)=t(\tau-I(t<0))$ with $0\leq \tau \leq 1$, then $\widehat \beta$ is the quantile regression estimator; see \cite{Koenker:1978}. On the other hand, if the density function $f$ is known, the maximum likelihood is the best approach to estimating $\beta$. Clearly the maximum likelihood function is
\begin{equation}\label{eq3}
  \prod_{i=1}^{n}f(Y_{i}-\nu-x_{i}^\top\beta)
\end{equation}

As we known, if the error term follows a normal distribution, the maximum likelihood estimator is  equivalent to the mean regression estimator, and if $f$ is the Laplace density, the maximum likelihood estimator is the same as the quantile regression estimator (\cite{Geraci:2006}). There is a huge body of literature about the semiparametric models. The semiparametric efficiency for the estimation of the parametric part can be established,  provided that the density of the error is symmetric about zero and absolutely continuous with respect to the Lebesgue measure (\cite{Bickel:1993}). However, in reality, we do not know how the distribution of the error term looks like. We might have to test whether it is symmetric or not, unimodal or multimodal. This motivates us to apply nonparametric techniques to estimate $f$ as a first step. After this we plug our estimated $f$ into \eqref{eq3} and find the minimizer $\widehat \beta$. This is the method we propose in the article. We call it the distribution regression, which is a distribution free regression and can be considered as generalizations of both mean regression and quantile regression.

We show that our distribution regression estimator is robust and possesses nice finite and large sample properties. The establishment of our asymptotic theories is very challenging because our nonparametric likelihood function cannot be transformed into sum of i.i.d. random variables. Numerical studies indicate that the distribution regression is much better compared to existing methods under many settings, particularly for asymmetrical heavy-tailed distribution or multimodal distribution of the error term. Our procedure essentially can be applied to other models, such as single or multiple index models, partial linear model and varying coefficient model. As a general robust approach, it has potentials to be used in testing the goodness of parameter estimation of mean regression.

Another popular and powerful tool in the linear regression is the penalty. This idea may date back to \cite{Tibshirani:1996}, who studied the mean regression via LASSO. Later on, the adaptive LASSO (\cite{Zou:2006}),  SCAD (\cite{Fan:2001}), MCP (\cite{Zhang:2010}) and many others were proposed. Other related works to our paper are: SCAD and adaptive LASSO in the quantile regression (\cite{Wu:2009}), local quadratic approximation (LQA, \cite{Fan:2001}) and local linear approximation (LLA, \cite{Zou:2008}).
 In this article, we combine adaptive LASSO with our method under the high dimensional setting and adopt LQA in the numerical estimation.  Under regular conditions, our penalized likelihood estimator has the so-called oracle property with diverging number of parameters. Furthermore, a new iterative coordinate descent algorithm is suggested for our estimator, which is an extended version of coordinate descent algorithm (\cite{Wu:2008}).

The rest of this article is organized as follows. In Section \ref{sec2}, we propose the distribution regression and establish its theoretical property. Sparse estimation in the high dimension setting is presented in Section \ref{sec3}. In Section \ref{sec4}, we examine the finite sample performance of the proposed method via Monte Carlo simulations, and also illustrate the effectiveness through several empirical examples. %The article concludes with a short discussion in Section \ref{sec5}.
All technical proofs of the main results are given in Appendix.

\section{Distribution Regression} \label{sec2}

 Given a set of i.i.d.  observations $\{z_{i}\}_{i=1}^{n}$, the classical kernel density estimator of  $f(z)$ is
$$
 \tilde f_{nh}(z)=\frac{1}{nh}\sum_{i=1}^{n}K\Big{(} \frac{z-z_{i}}{h}   \Big{)},
$$
where $K(\bullet)$ is the kernel function, and $h$ is the bandwidth. Technically it may happen that $\tilde f_{nh}(z)=0$ if $f(z)$ is close to $0$. In order to avoid such annoying situation, we define the kernel density estimator of  $f(z)$ in this article as
 \begin{equation}\label{eq4}
  f_{nh}(z)=\frac{1}{nh}\sum_{i=1}^{n}K\Big{(} \frac{z-z_{i}}{h}   \Big{)}+n^{-1000}.
\end{equation}

 The reason to have the extra term $n^{-1000}$ is that we have to deal with log-likelihood function in the theoretical development. In the numerical studies, this term is ignorable. The nonparametric likelihood function is given by

 %So the nonparametric likelihood function is
\begin{equation}\label{eq5}
L_{nh}(\beta)= \prod_{j=1}^{n}f_{nh}(\varepsilon_{j}).
\end{equation}
The  estimator $\widehat{\beta}$ can be found by maximizing \eqref{eq5}. Since it is based on the estimated joint distribution of error terms, we call $\widehat{\beta}$ the distribution regression estimator.
For simplicity, we denote the log-likelihood function $\log L_{nh}$ by $\ell_{nh}$.

Combining (\ref{eq4}) and (\ref{eq5}), we have

\begin{align}\label{eq6}
   &\ell_{nh}(\beta)=\sum_{j=1}^{n} \log\Big{[} \frac{1}{nh}\sum_{i=1}^{n}K\Big{(} \frac{\varepsilon_{j}-\varepsilon_{i}}{h}   \Big{)}  +n^{-1000} \Big{]} \notag \\
   & =  \sum_{j=1}^{n}\log \Big{[} \frac{1}{nh}\sum_{i=1}^{n}K\Big{(} \frac{(Y_{j}-\nu-x_{j}^\top\beta)-(Y_{i}-\nu-x_{i}^\top\beta)}{h}   \Big{)}  +n^{-1000} \Big{]}  \notag \\    %\nonumber
   & =  \sum_{j=1}^{n}\log \Big{[} \frac{1}{nh}\sum_{i=1}^{n} K\Big{(} \frac{(Y_{j}-Y_{i})- (x_{j}-x_{i})^\top \beta }{h}   \Big{)}  +n^{-1000} \Big{]} %\tag{6}
\end{align}

Clearly, (\ref{eq6}) is a nonlinear objective function. This is more complicated to handle than other ones. However, maximizing the nonlinear function as above is not as difficult as it looks. In the article, we adopt the algorithm of \cite{Varadhan:2009}, which optimizes a high-dimensional nonlinear objective function, and achieves fast computation.

\begin{remark}
The bandwidth $h$ is a tuning parameter.  Many bandwidth selection methods can be used here since our estimate is essentially the same as the usual kernel estimation. In the simulations, we use the plug-in method of \cite{Wand:1994} to select the bandwidth, which performs well in our numerical studies.
\end{remark}

\begin{remark}
The intercept term $\nu$ of regression model  disappears in (\ref{eq6}). So our method can't directly provide one estimator of the intercept term. However one can use other methods, such as mean regression, to get the estimator of the intercept term after obtaining $\widehat{\beta}$. From now on, we assume that $\nu=0$ in our model. This assumption doesn't have any influence on $\widehat \beta$.
\end{remark}

%\subsection{Asymptotic property }
In order to state our main results, we need the following notation.  Denote by $\Omega$ the parameter space of $\beta$, $\beta_{0}$ the true value of $\beta$, and $\widehat{\beta}$ the estimator of $\beta_{0}$, $B(\kappa)=\{ \beta: \|\beta-\beta_{0}\|\leq \kappa\}$, a ball centered at $\beta_0$ with radius $\kappa$. Its boundary will be denoted by $\partial B(\kappa)$. Set $U=(x,Y)$, $U_{i}=(x_{i}, Y_{i}), i=1,\dots,n$. Let $\ell(\beta)$ be the log-likelihood function of the observations $U_{1}, \dots, U_{n}$. To facilitate asymptotic properties of our proposed estimator, we assume the following conditions.
\begin{condition}
\label{condition1}
  $$
 \|K\|_\infty:= \sup_{z\in R} | K(z)| < \infty , ~~\int_{z\in R}  | K(z)|^2 dz < \infty,~~ \lim_{|z|\rightarrow \infty }|zK(z)| =0;
  $$
  The first derivative $K'(z)$ is piecewise continuous and
  $$\|K'\|_\infty:=\inf \{C\geq 0: | K'(z)| \leq C \ \text{for almost all} \ z\in R\}<\infty.$$
\end{condition}

\begin{condition}
\label{condition2}
  There exists a positive constant $\alpha<1/2$ such that $ n^{1-\alpha}h/\log n\to \infty$ and $n^{\alpha}h^2\to 0$ as $n\to \infty$.
%the 2nd condition is used in the proof of Lemma \ref{lemma2} when replacing $f(\varepsilon)$ by $EZ_i$.
\end{condition}

\begin{condition}
\label{condition3}
 The probability density function $f(z)$ of the error term is uniformly continuous; the first and second derivatives $f'(z)$, $f''(z)$ exist a.e., and $\| f'\| _\infty< \infty, \|f''\|_\infty < \infty$; there exist constants $\kappa_0$ and $k\geq 1$  such that  $\sup_{p\geq 1}P_{\beta_0}(\min_{\beta\in B(\kappa_0)}f(Y-x^\top \beta)\leq n^{-k})=o(1/n)$ and $\sup_{p\geq 1}P_{\beta_0}(\min_{\beta\in B(\kappa_0)}$ $f(Y-x^\top \beta)\leq n^{-\alpha})=o(1/\log n)$ as $n\to \infty$. (Here $P_{\beta_0}$ means the true parameter is $\beta_0$. So does $E_{\beta_0}$.)
 \end{condition}

\begin{condition}
\label{condition4}
 $U_{1},\dots, U_{n}$ are i.i.d. with probability density $f(U,\beta)$ satisfying $\sup_{1\leq p<\infty}E\|x/p\|<\infty$.  Note that $f(U,\beta)=f(Y-x^\top\beta)h(x)$, where $h(x)$ is the marginal density of $x$.  The model is identifiable, i.e. $\forall \beta \neq \beta', \{ U: f(U, \beta) \neq f(U,\beta')\}$ is not a Null-set. There exists an open subset $\omega$ of $\Omega$ that contains the true parameter point $\beta_{0}$  for almost all $U$. Moreover,  $\sup_{\beta\in B(\kappa_0),p\geq 1}E_{\beta_0}|\log f(Y-x^\top\beta)|<\infty$, and for any $\kappa\in (0, \kappa_0]$
 \begin{equation} \label{C24-1}
\sup_{p\geq 1} \sup_{\beta\in \partial B(\kappa)}E_{\beta_0}\log\frac{ f(Y-x^\top\beta)}{ f(Y-x^\top\beta_0)}<0.
 \end{equation}
%  , and $f(U ,\beta)$ admits all third derivatives $\partial^{3}/\partial\beta_{j}\partial\beta_{k}\partial\beta_{l} f(U,\beta)$ for all $\beta \in \omega$.
\end{condition}

\begin{condition}
\label{condition5}
For almost all observed $(Y,x)$, there exists a function $M(U)$ such that
$$
\sup_{\beta\in B(\kappa_0)}\|\frac{\partial}{\partial Y}\log f(Y-x^\top\beta)\| \leq M(U),
$$
where $\sup_{p\geq 1}E\big(M(U)\|x/\sqrt p\|\big)<\infty$. If $p\to \infty$,
$$
 \sup_{p\geq 1}E\big(M(U)\|x/\sqrt p\|\big)^3<\infty.
$$
\end{condition}

\begin{remark}
Conditions \ref{condition1}--\ref{condition2} are regular in the literature of kernel density estimator, which are used to prove the consistency of the kernel density estimator. Condition \ref{condition3} is actually very mild. Almost all the commonly seen distributions, such as (mixed) normal,  double exponential, Cauchy distributions, satisfy it.  Conditions \ref{condition4}  requires that  $E_{\beta_0}|\log f(Y-x^\top\beta)|<\infty$ if $\beta\in B(\kappa_0)$. This implies that $Y-x^\top\beta$ is $a.s.$ in the support of  $\varepsilon$. So the support of $\varepsilon$ should be the whole real line.% Conditions \ref{condition4} and \ref{condition5} are weaker than the usual ones to guarantee consistency of the ordinary maximum likelihood estimators (\cite{Lehmann:1998}) if $p$ is fixed. Actually \eqref{C24-1} is automatically correct by Jensen's inequality..
\end{remark}

\begin{theorem}
\label{theorem1}
Let $U_{1},\dots, U_{n}$ be i.i.d. from one population with the density function $f(U,\beta)$.
\begin{enumerate}
\item
Suppose $p$ is fixed. Under Conditions \ref{condition1}--\ref{condition5},  $\widehat{\beta}\to \beta_{0}$ in probability as $n\to \infty$.
\item
Suppose $p=o(n^{1-\alpha}h/\log n)$ as $n\to \infty$, and
\begin{equation} \label{momentgen}
\sup_{\beta\in B(\kappa_0),p\geq 1} E_{\beta_0}\exp\Big(\psi|\log f(Y-x^\top\beta)|\Big)<\infty
 \end{equation}
 for some positive $\psi$. Under Conditions \ref{condition1}--\ref{condition5},  $\widehat{\beta}-\beta_{0}\to 0$ in probability as $n\to \infty$.
\end{enumerate}
\end{theorem}

Theorem \ref{theorem1} shows that the distribution regression estimator is consistent  under regular conditions.

In order to derive the central limit theorem for $\widehat \beta$, we need more conditions. Let $C(\kappa)=\{ \beta: \|\beta-\beta_{0}\|\leq \kappa\}$, a complex ball in ${\bf C}^p$ centered at $\beta_0$ with radius $\kappa$. Here each coordinate of $\beta$ is complex. We also need to extend the domain of $f(z)$ from the real line to the complex plane. Here we take the double exponential distribution as one example to see how to define $f(z)$ in the complex plane. In this case, $f(z)=e^{-|z|}/2$. For complex $\beta$, we define
$$
f(Y-x^\top\beta)=
\left\{
\begin{array}{cl}
e^{-(Y-x^\top\beta)} /2     &   \text{if} ~~ \Re(Y-x^\top \beta)\geq 0, \\
e^{(Y-x^\top\beta)} /2    &   \text{if} ~~ \Re(Y-x^\top \beta)< 0.
\end{array} \right.
$$

%Define $A_{\kappa}=\{f(Y-x^\top \beta)$ is an analytic function of each coordinate of $\beta$ in $C(\kappa) \}$.
%\begin{condition}
%\label{condition6}
%$K(z)$ is an analytic function in the complex plane. For any $\epsilon>0$, there exists a positive $\kappa$ such that
%$$
%P(A_\kappa)\geq 1-\epsilon.
%$$
%\end{condition}
%\begin{remark}
%All the commonly seen distributions satisfy the second part of Condition \ref{condition6}. Let us take the exponential distribution as one example. In this case the density function is $f(z)=e^{-z}, z>0$. So for complex $\beta$, we have
%$$
%f(Y-x^\top\beta)=
%\left\{
%\begin{array}{cl}
%e^{-(Y-x^\top\beta)}      &   \text{if} ~~ \Re(Y-x^\top \beta)\geq 0, \\
%0    & \text{otherwise}.
%\end{array} \right.
%$$
% Note that $P\big(Y-x^\top\beta_0>-\log(\epsilon/2) \big)=\epsilon/2$ for any positive $\epsilon<1$. Let $C_0$ be any constant such that $P(\|x\|>C_0)\leq \epsilon/2$.  Then we can set $\kappa$ to be $-\log(\epsilon/2)/C_0$ since $\|x^\top(\beta-\beta_0)\| \leq -\log(\epsilon/2)$ when $\|x\|\leq C_0$ and $\beta\in C(\kappa)$.
%    Under this condition, we can provide a short proof of the central limit theorem for $\widehat \beta$. We believe that this analytic condition could be removed.
%\end{remark}

\begin{condition}
\label{condition6}
There exist constants $\kappa_0>0$ and $k\geq 1$ such that  $P(\min_{\beta\in C(\kappa_0)}$ $\|f(Y-x^\top \beta)\|\leq n^{-k})=o(1/n)$ and $P(\min_{\beta\in C(\kappa_0)}\|f(Y-x^\top \beta)\|\leq n^{-\alpha})=o(1/\log n)$ as $n\to \infty$,  and $E_{\beta_0}\log f(Y-x^\top\beta)$ is an analytic function of each coordinate of $\beta$ in $C(\kappa_0)$; for any $\beta\in C(\kappa_0)$, $f'(Y-x^\top\beta)$ and $f''(Y-x^\top\beta)$ exist for almost all observed $(Y,x)$; $\sup_{\beta\in C(\kappa_0)} \|f'(Y-x^\top\beta)\|_\infty<\infty$, $\sup_{\beta\in C(\kappa_0)} \|f''(Y-x^\top\beta)\|_\infty<\infty$, where
\begin{eqnarray*}
\|f'(Y-x^\top\beta)\|_\infty=\inf\{C>0:P(\|f'(Y-x^\top\beta)\|\leq C)=1\},\\
\|f''(Y-x^\top\beta)\|_\infty=\inf\{C>0:P(\|f''(Y-x^\top\beta)\|\leq C)=1\};
\end{eqnarray*}
 for almost all observed $(Y,x)$, $\log f(Y-x^\top\beta)$ is an analytic function of each coordinate of $\beta$ in $C(\kappa)$ if $\kappa$, which may depend on $(Y,x)$, is sufficiently small and positive.
\end{condition}

\begin{condition}
\label{condition7}
 The first and second derivatives of $\log f(Y-x^\top\beta)$ satisfy the equations
  $$
  E_{\beta_0}\Big{[}   \frac{\partial}{\partial b_{j}}\log f(Y-x^\top\beta_0)   \Big{]}
  =\frac{\partial}{\partial b_{j}}E_{\beta_0}\Big{[}   \log f(Y-x^\top\beta_0)   \Big{]} =0 ~~j=1, \dots, p,
  $$
and
\begin{align*}
\bm{I}_{jk}(\beta_0)& = E_{\beta_0}\Big{[}   \frac{\partial}{\partial b_{j}}\log f(Y-x^\top\beta_0)\cdot \frac{\partial}{\partial b_{k}} \log f(Y-x^\top\beta_0) \Big{]}\\
  &=  E_{\beta_0}\Big{[} -  \frac{\partial^{2}}{\partial b_{j}\partial b_{k}} f(Y-x^\top\beta_0) \Big{]}\\
  &=\frac{\partial^{2}}{\partial b_{j}\partial b_{k}} E_{\beta_0}\Big{[} -   f(Y-x^\top\beta_0) \Big{]}.
\end{align*}
The Fisher information matrix
$$
\bm{I}(\beta) = E_{\beta}\Big{\{} \Big{[} \frac{\partial}{\partial \beta }\log f(Y-x^\top\beta) \Big{]} \Big{[}   \frac{\partial}{\partial \beta }\log f(Y-x^\top\beta) \Big{]}^\top \Big{\}}
$$
is finite and positive definite at $\beta=\beta_{0}$.
\end{condition}

\begin{condition}
\label{condition8}
 $\sup_{\beta\in C(\kappa_0),p\geq 1}E_{\beta_0}|\log f(Y-x^\top\beta)|<\infty$.
Condition \ref{condition5} still holds if one replaces $B(\kappa)$ by $C(\kappa_0)$.
\end{condition}

\begin{theorem}
\label{theorem2}
Let $U_{1},\dots, U_{n}$ be i.i.d. from one population with the density function $f(U,\beta)$.
Under Conditions \ref{condition1}--\ref{condition8}, if $p$ is fixed, $\sqrt{n} (\widehat{\beta} -\beta_{0}) $ is asymptotically normal with mean zero and covariance matrix $[\bm{I}(\beta_{0})]^{-1}$ .
\end{theorem}

\begin{remark}
Under Conditions \ref{condition1}--\ref{condition8}, Theorem \ref{theorem2} tells us that $\widehat{\beta}$ is not only $\sqrt n$ consistent, but also asymptotically efficient.
\end{remark}

\section{Sparse Estimation in High Dimension} \label{sec3}

When the true underlying model has a sparse representation, meaning that most of the true unknown parameters are zero, variable selection is necessary. In this section, we study the penalized likelihood estimator via adaptive LASSO. In other words, we find our sparse estimator by minimizing the penalized likelihood function
\begin{equation}\label{eq7}
-\frac{1}{n} \ell_{nh}(\beta) + \sum_{j=1}^{p}p_{\lambda_{n}}(|b_{j}|).
\end{equation}
where $p_{\lambda}(\bullet)$ is the penalty function, which includes four popular choices: LASSO (\cite{Tibshirani:1996}), adaptive LASSO (\cite{Zou:2006}),  SCAD (\cite{Fan:2001}) and MCP (\cite{Zhang:2010}). In this article, we use adaptive LASSO. So, (\ref{eq7}) can be written as
\begin{equation}\label{eq8}
 -\frac{1}{n} \ell_{nh}(\beta) + \lambda_{n} \sum_{j=1}^{p}\hat{w}_{j}|b_{j}|.
\end{equation}
where $\hat{w}_{j}=1/|\hat{b}_{j}|^{\gamma}, \gamma >0$, $\hat b_j$ is the our distribution regression estimator of the $j$th coordinate of $\beta$ without any penalty.

To solve the minimization problem in (\ref{eq8}), we implement iterative coordinate descent algorithm combined with local quadratic approximation.
 Given $b_{1}^{(k)},\dots, b_{p}^{(k)}$, we can find
\begin{equation}\label{eq9}
b_{j}^{(k+1)} = \arg\min_{b_{j}}\Big{\{}    -\frac{1}{n} \ell_{nh}(b_{j}) + \lambda_{n} \frac{1}{2[b_{j}^{(k)}]^{2}}b_{j}^2   \Big{\}}
\end{equation}
 %\arg\min\Big{\{}    -\frac{1}{n} \ell_{nh}(b_{1}^{(k)},\cdots,b_{j-1}^{(k)},b_{j}, b_{j+1}^{(k)},\cdots,b_{p}^{(k)} ) \notag \\
 %&~~~~+ \lambda_{n} \Big{(}  \frac{1}{2[b_{1}^{(k-1)}]^{2}}[b_{1}^{(k)}]^{2} + \dots  +   \frac{1}{2[b_{j-1}^{(k-1)}]^{2}}[b_{j-1}^{(k)}]^{2} \notag\\
 %&~~~~+\frac{1}{2[b_{j}^{(k)}]^{2}}b_{j}^2  \notag \\
 %&~~~~+  \frac{1}{2[b_{j+1}^{(k-1)}]^{2}}[b_{j+1}^{(k)}]^{2}  + \dots +  \frac{1}{2[b_{p}^{(k-1)}]^{2}}[b_{p}^{(k)}]^{2} \Big{)}     \Big{\}} \notag\\

An iterative coordinate descent algorithm is as follow.
\begin{itemize}
  \item[Step 1] Input an initial value $\beta^{(0)}$
  \item[ Step 2] For $k\geq 0$
  \begin{itemize}
    \item[Step 2.1] For $j\in {1,\dots,p}$,  update $b_{j}^{(k+1)}$ by (\ref{eq9})
    \item[Step 2.2] Repeat Step 2.1, until $b_{j}^{(k+1)}$ converges.
  \end{itemize}
  \item[Step 3]  Update $k \leftarrow k+1$. Repeat  Step 2 until convergence. Then output $\beta^{(k+1)}$.
\end{itemize}

The selection of the tuning parameter $\lambda_n$ is important in real applications. Cross-validation is a very common method, however it is time-consuming when $p$ is large. In the article, we use the generalized information criterion (GIC) for high dimensional penalized likelihood (\cite{Fan:2013, Lin:2014}). The GIC is defined by
$$
GIC(\lambda_{n})= \log \widehat{\sigma}^{2} + s_{\lambda_{n}} \frac{\log \log n}{n} \log(p \vee n)
$$
where $\widehat{\sigma}^{2}=\sum_{i=1}^n(Y_i-x_i^\top\widehat{\beta}_{\lambda_{n}})^{2} / n$, $\widehat{\beta}_{\lambda_{n}}$ is the estimator which minimizes \eqref{eq8}, $p\vee n=\max(p,n) $, and $s_{\lambda_{n}}$ is the number of nonzero coefficients in $\widehat{\beta}_{\lambda_{n}}$. So $\lambda_n$ is selected as
$$\hat{\lambda}_{n}=\arg\min_{\lambda_{n}}GIC(\lambda_{n}).$$

Now we study the large sample properties of  the proposed penalized likelihood estimators. Let $x_{i}=(x_{i1}^\top, x_{i2}^\top )^\top$, $x_{i1}\in R^{s}, x_{i2}\in R^{p-s},$ $\beta=(\beta_{1}^\top, \beta_{2}^\top )^\top$. The true nonzero coefficients are $\beta_{10}$, and  the zero coefficients are $\beta_{20}$. The corresponding estimator is $\widehat{\beta}_{0}=(\widehat{\beta}_{10}^\top, \widehat{\beta}_{20}^\top )^\top$. We rewrite (\ref{eq8}) as
\begin{equation}\label{eq10}
Q(\beta)= - \ell_{nh}(\beta) + n\lambda_{n} \sum_{j=1}^{p}\hat{w}_{j}|b_{j}|.
\end{equation}

\begin{theorem}[Oracle properties]
\label{theorem3}
In addition to the conditions in Theorem \ref{theorem2}, we assume that  $\liminf_{n\to \infty}\lambda_{n}\sqrt{n}>0$. If $p\to \infty$ and $s$ is fixed, we further assume that $p=O(\lambda_n n^{1/2})$, $p=o(n^{1-\alpha}h/\log n)$ as $n\to \infty$, $\sup_{p\geq 1} \|\bm{I}(\beta_0)\|<\infty$, and $\sup_{\beta\in C(\kappa_0),p\geq 1} E_{\beta_0}\exp\big(\psi|\log(Y-x^\top\beta)|\big)<\infty$ for some positive $\psi$. Then we have

(a) Sparsity: $\widehat{\beta}_{20}=0$ with probability tending to one as $n\to \infty$.

(b) Asymptotic normality: $\sqrt{n}(\widehat{\beta}_{10} -\beta_{10}) \xrightarrow{d} N(0,[\bm{I}_{11}(\beta_{0})]^{-1})$ , where
$$
\bm{I}_{11}(\beta_{0})=E_{\beta_0}\Big{\{} \Big{[} \frac{\partial}{\partial \beta_1 }\log f(Y-x^\top\beta_0) \Big{]} \Big{[}   \frac{\partial}{\partial \beta_1 }\log f(Y-x^\top\beta_0) \Big{]}^\top \Big{\}}.
$$
\end{theorem}
Theorem \ref{theorem3} states that the distribution regression with adaptive LASSO estimator have oracle properties in the sparse linear models with diverging number of parameters.

\section{Simulation and application}\label{sec4}
\subsection{Simulation}
In this section, we evaluate the finite sample performance of our proposed distribution regression  with two existing methods (mean regression and quantile (median) regression) through
two sets of simulation examples.  We set $n$ to be 100, and repeat 200 times. $x=(x_{1},\dots,x_{p})\sim N(0,E)$, where  $E$ is a $p\times p$ identity matrix. The error term $\varepsilon$ is generated from seven distributions: (1) the mixed normal distribution: $0.5N(-2, 5^2)+0.5N(2, 0.5^2)$; (2) the normal distribution $N(0,1)$; (3) $t(3)$;  (4) the standard Laplace distribution: $Lap(0,1)$; (5) the Gamma distribution $G(2,2)$; (6) the standard Cauchy distribution $Cau(0,1)$; (7) the exponential distribution: $Exp(1)$. Among these, $t(3)$ and $Cau(0,1)$ are symmetric heavy-tailed distributions; the mixed normal distribution is multimodal; $G(2,2)$ and $Exp(1)$ are asymmetric.

\begin{example}
\label{example1}
In this example, we consider three linear models:
\begin{itemize}
  \item[] Case 1: $Y= 3x_{1}+3x_{2}+3x_{3}+\varepsilon $
  \item[] Case 2: $  Y= (3-Q_{1})x_{1}+(3-Q_{2})x_{2}+(3-Q_{3})x_{3}+\varepsilon $
  \item[] Case 3:  $  Y= (3-Q_{1})x_{1}+\dots+(3-Q_{6})x_{6}+\varepsilon$
\end{itemize}
where $Q_{j}$ is generated from a uniform distribution on $[0,1]$. We evaluate the performance through the following two criteria:
\begin{enumerate}[(1)]
  \item $bias_{j}=\frac{1}{200}\sum_{m=1}^{200}|\hat{b}_{j}^{(m)}-b_{j}|,~~~j=1\dots,p$
  \item $MSE =\frac{1}{200p}\sum_{m=1}^{200}\sum_{j=1}^{p}(\hat{b}_{j}^{(m)}-b_{j})^2$.
\end{enumerate}
Here $\hat{b}_{j}^{(m)}$ is the estimator of $b_j$ based on the $m$th sample.
\end{example}

Tables \ref{Tab1} and \ref{Tab2} summarize the simulation results. We can draw the following conclusions:
\begin{enumerate}
  \item When the error term follows the normal distribution, the mean regression (MR) performs the best and the distribution regression (DR) has almost the same performance as the quantile regression (QR). However, the difference of their performance is not big. When the error term follows other distributions, the mean regression becomes worse as the population is moving away from the normal.
  \item When the error follows the standard Cauchy distribution, the quantile regression is the best.
  \item When the error term is asymmetric or heavy-tailed or multimodal, the performance of  our distribution regression is the best and the performance of mean regression is the worst.
\end{enumerate}

\begin{table*}
\centering
\caption{The simulation results of Cases 1 and 2.}
\label{Tab1}
\begin{tabular}{crrr rrrc}
\hline
 	&	     &\multicolumn{3}{c}{Case 1} & \multicolumn{3}{c}{Case 2 } \\
	&		&	DR	&	MR	&	QR	&	DR	&	MR	&	QR	\\\hline
Mixed Norm	&	$bias_{1}$	&	0.1271	&	0.3331	&	0.5126	&	0.1335	&	0.3649	&	0.5293	\\
	&	$bias_2$	&0.1210  &	0.3365	&	0.5126	&	0.1267	&	0.3119	&	0.4695	\\
	&	$bias_3$	&	0.1262	&	0.3139	&	0.5135	&	0.1256	&	0.3502	&	0.5270	\\
	&	$MSE$	&	0.0375	&	0.1696	&	0.3696	&	0.0424	&	0.1853	&	0.3682	\\
$N(0,1)$	&	$bias_{1}$	&	0.1056	&	0.0812	&	0.1040	&	0.1032	&	0.0775	&	0.1004	\\
	&	$bias_2$	&	0.1026	&	0.0789	&	0.1054	&	0.0995	&	0.0757	&	0.0997	\\
	&	$bias_3$	&	0.1059	&	0.0824	&	0.1035	&	0.1079	&	0.0814	&	0.1019	\\
	&	$MSE$	&	0.0171	&	0.0103	&	0.0167	&	0.0169	&	0.0097	&	0.0159	\\
t(3)	&	$bias_{1}$	&	0.1370	&	0.1332	&	0.1093	&	0.1397	&	0.1322	&	0.1155	\\
	&	$bias_2$	&	0.1445	&	0.1314	&	0.1163	&	0.1428	&	0.1350	&	0.1123	\\
	&	$bias_3$	&	0.1358	&	0.1438	&	0.1117	&	0.1444	&	0.1433	&	0.1159	\\
	&	$MSE$	&	0.0310	&	0.0309	&	0.0203	&	0.0318	&	0.0391	&	0.0208	\\
$Lap(0,1)$ 	&	$bias_{1}$	&	0.1209	&	0.1191	&	0.0921	&	0.1125	&	0.1133	&	0.0964	\\
	&	$bias_2$	&	0.1168	&	0.1118	&	0.0941	&	0.1216	&	0.1100	&	0.0956	\\
	&	$bias_3$	&	0.1152	&	0.1151	&	0.0955	&	0.1141	&	0.1124	&	0.0950	\\
	&	$MSE$	&	0.0223	&	0.2093	&	0.0142	&	0.0219	&	0.0203	&	0.0147	\\
$G(2,2)$	&	$bias_{1}$	&	0.0520	&	0.0943	&	0.1533	&	0.0520	&	0.0914	&	0.1601	\\
	&	$bias_2$	&	0.0532	&	0.0995	&	0.1600	&	0.0509	&	0.0982	&	0.1548	\\
	&	$bias_3$	&	0.0545	&	0.1033	&	0.1648	&	0.0528	&	0.0968	&	0.1567	\\
	&	$MSE$	&	0.0050	&	0.0154	&	0.0366	&	0.0048	&	0.0142	&	0.0358	\\
$Cau(0,1)$	&	$bias_{1}$	&	1.2789	&	3.3662	&	0.1300	&	1.1963	&	3.8744	&	0.1270	\\
	&	$bias_2$	&	0.8949	&	14.561	&	0.1410	&	1.3502	&	4.0455	&	0.1349	\\
	&	$bias_3$	&	0.9317	&	3.4846	&	0.1363	&	1.2723	&	4.3668	&	0.1429	\\
	&	$MSE$	&	52.938	&	24068	&	0.0301	&	69.329	&	310.40	&	0.0304	\\
$Exp(1)$	&	$bias_{1}$	&	0.0460	&	0.1118	&	0.0885	&	0.0455	&	0.1140	&	0.0907	\\
	&	$bias_2$	&	0.0475	&	0.1180	&	0.0949	&	0.0437	&	0.1196	&	0.0895	\\
	&	$bias_3$	&	0.0489	&	0.1145	&	0.0908	&	0.0472	&	0.1100	&	0.0890	\\
	&	$MSE$	&	0.0050	&	0.0213	&	0.0136	&	0.0044	&	0.0202	&	0.0127	\\\hline
\end{tabular}
\end{table*}

\begin{table*}
\centering
\caption{The simulation results of Case 3.}
\label{Tab2}
\begin{tabular}{crrr rrrr c}
\hline
 	&		&	$bias_1$	&	$bias_2$	&	$bias_3$	&	$bias_4$	&	$bias_5$	&	$bias_6$	&	$MSE$	\\\hline
Mixed Norm	&	DR	&	0.1742	&	0.1632	&	0.1689	&	0.1551	&	0.1643	&	0.1519	&	0.0602	\\
	&	MR	&	0.3446	&	0.3461	&	0.3339	&	0.3196	&	0.3268	&	0.3312	&	0.1790	\\
	&	QR	&	0.4622	&	0.4456	&	0.4369	&	0.4307	&	0.4465	&	0.4503	&	0.3016	\\
$N(0,1)$	&	DR	&	0.1021	&	0.1074	&	0.1097	&	0.1086	&	0.1027	&	0.1081	&	0.0182	\\
	&	MR	&	0.0803	&	0.0819	&	0.0860	&	0.0818	&	0.0809	&	0.0819	&	0.0151	\\
	&	QR	&	0.0987	&	0.0993	&	0.1095	&	0.1060	&	0.1056	&	0.1079	&	0.0171	\\
$t(3)$	&	DR	&	0.1406	&	0.1335	&	0.1382	&	0.1350	&	0.1336	&	0.1300	&	0.0296	\\
	&	MR	&	0.1405	&	0.1300	&	0.1340	&	0.1356	&	0.1383	&	0.1387	&	0.0309	\\
	&	QR	&	0.1159	&	0.1150	&	0.1156	&	0.1163	&	0.1172	&	0.1167	&	0.0216	\\
$Lap(0,1)$&	DR	&	0.1217	&	0.1160	&	0.1280	&	0.1202	&	0.1188	&	0.1212	&	0.0235	\\
	&	MR	&	0.1101	&	0.1185	&	0.1181	&	0.1137	&	0.1154	&	0.1161	&	0.0214	\\
	&	QR	&	0.0964	&	0.0993	&	0.1037	&	0.0948	&	0.0997	&	0.0962	&	0.0159	\\
$G(2,2)$	&	DR	&	0.0621	&	0.0560	&	0.0606	&	0.0633	&	0.0612	&	0.0595	&	0.0065	\\
	&	MR	&	0.0969	&	0.0977	&	0.1042	&	0.1039	&	0.1016	&	0.1008	&	0.0161	\\
	&	QR	&	0.1445	&	0.1420	&	0.1476	&	0.1511	&	0.1441	&	0.1404	&	0.0323	\\
$Cau(0,1)$&	DR	&	2.1296	&	1.6145	&	2.0975	&	1.6555	&	1.6499	&	3.2405	&	269.04	\\
	&	MR	&	5.5058	&	3.1172	&	4.0740	&	3.8904	&	3.4932	&	4.4697	&	555.88	\\
	&	QR	&	0.1404	&	0.1429	&	0.1415	&	0.1482	&	0.1522	&	0.1470	&	0.0346	\\
$Exp(1)$	&	DR	&	0.0643	&	0.0607	&	0.0565	&	0.0591	&	0.0631	&	0.0605	&	0.0075	\\
	&	MR	&	0.1196	&	0.1130	&	0.1117	&	0.1188	&	0.1242	&	0.1144	&	0.0220	\\
	&	QR	&	0.0994	&	0.0973	&	0.0934	&	0.0989	&	0.1004	&	0.0956	&	0.0153	\\\hline
\end{tabular}
\end{table*}

\begin{example}
\label{example2}
 In this example, we study the performance of variable selection. Consider the following linear model:
$$
Y=3x_{1}+1.5x_{2}+2x_{3} + 0x_{4}+\dots +0x_{p}+\varepsilon
$$
We set the dimension $p$ as 10 and 50.  Besides $bias$, we also use the following criteria to evaluate the performance:
\begin{enumerate}[(1)]

  \item Correct-fit: the proportion of including all three significant variables.
  \item Over-fit: the proportion of including all three significant variables and some zero coefficients.
  \item Under-fit: the proportion of excluding any significant variables.
  \item  $ME$ = $(\widehat{\beta}- \beta_{0})^\top E(xx^\top)(\widehat{\beta}- \beta_{0})$
  \item C: the average number of nonzero coefficients correctly estimated to be nonzero.
  \item IC: the average number of zero coefficient incorrectly estimated to be nonzero.
\end{enumerate}
\end{example}

Tables \ref{Tab3} and \ref{Tab4} summarize the simulation results. Under the criterion $ME$, the mean regression performs well in the case of $N(0,1)$ , and the quantile regression performs well in cases of $t(3), Lap(0,1) $ and $Cau(0,1)$. However, for the other distributions, the distribution regression is superior to the mean regression and quantile regression, especially for mixed normal distribution and $Exp(1)$.  Under the criterion of variable selection, distribution regression outperforms mean and quantile regression.

\begin{table*}
\centering
\caption{The simulation results of Example \ref{example2} with $p=10$.}
\label{Tab3}
\begin{tabular}{crrr rrrr rrc}
\hline
	&		&	$bias_{1}$	&	$bias_{2}$	&	$bias_{3}$	&	Correct	&	Over	&	Under	&	ME	&		C	&	IC	\\\hline
Mixed Norm	&	DR	&	0.121	&	0.188	&	0.143	&	0.945	&	0.000	&	0.055	&	0.279	&		2.925	&	0.000	\\
	&	MR	&	0.367	&	0.426	&	0.392	&	0.450	&	0.500	&	0.050	&	1.092	&		2.945	&	1.090	\\
	&	QR	&	0.168	&	0.176	&	0.176	&	0.450	&	0.545	&	0.005	&	0.313	&		2.995	&	0.900	\\
$N(0,1)$	&	DR	&	0.111	&	0.104	&	0.112	&	0.960	&	0.040	&	0.000	&	0.055	&	3.000	&	0.045	\\
	&	MR	&	0.085	&	0.082	&	0.088	&	0.865	&	0.135	&	0.000	&	0.038	&	3.000	&	0.220	\\
	&	QR	&	0.103	&	0.094	&	0.101	&	0.745	&	0.255	&	0.000	&	0.056	&	3.000	&	0.285	\\
$t(3)$	&	DR	&	0.137	&	0.161	&	0.135	&	0.965	&	0.035	&	0.000	&	0.098	&	3.000	&	0.035	\\
	&	MR	&	0.144	&	0.147	&	0.148	&	0.690	&	0.310	&	0.000	&	0.138	&	3.000	&	0.505	\\
	&	QR	&	0.114	&	0.113	&	0.110	&	0.645	&	0.355	&	0.000	&	0.079	&	3.000	&	0.435	\\
$Lap(0,1)$	&	DR	&	0.121	&	0.128	&	0.118	&	0.980	&	0.020	&	0.000	&	0.070	&	3.000	&	0.020	\\
	&	MR	&	0.120	&	0.121	&	0.116	&	0.755	&	0.245	&	0.000	&	0.087	&	3.000	&	0.535	\\
	&	QR	&	0.100	&	0.096	&	0.095	&	0.760	&	0.240	&	0.000	&	0.055	&	3.000	&	0.270	\\
$G(2,2)$	&	DR	&	0.049	&	0.050	&	0.046	&	0.995	&	0.005	&	0.000	&	0.011	&	3.000	&	0.005	\\
	&	MR	&	0.065	&	0.062	&	0.058	&	0.975	&	0.025	&	0.000	&	0.018	&	3.000	&	0.070	\\
	&	QR	&	0.070	&	0.068	&	0.062	&	0.950	&	0.050	&	0.000	&	0.020	&	3.000	&	0.060	\\
$Cau(0,1)$	&	DR	&	1.389	&	0.975	&	1.161	&	0.410	&	0.000	&	0.590	&	7.082	&	1.495	&	0.000	\\
	&	MR	&	1.678	&	0.113	&	1.353	&	0.135	&	0.225	&	0.640	&	10.482	&	1.555	&	0.790	\\
	&	QR	&	0.138	&	0.141	&	0.146	&	0.560	&	0.440	&	0.000	&	0.137	&	3.000	&	0.575	\\
$Exp(1)$	&	DR	&	0.033	&	0.042	&	0.039	&	1.000	&	0.000	&	0.000	&	0.007	&	3.000	&	0.000	\\
	&	MR	&	0.076	&	0.092	&	0.083	&	0.830	&	0.170	&	0.000	&	0.037	&	3.000	&	0.250	\\
	&	QR	&	0.074	&	0.075	&	0.079	&	0.805	&	0.095	&	0.000	&	0.028	&	3.000	&	0.110	\\\hline
\end{tabular}
\end{table*}

\begin{table*}
\centering
\caption{The simulation results of Example \ref{example2} with $p=50$.}
\label{Tab4}
\begin{tabular}{crrr rrrr rrc}
\hline
  	&		&	$bias_{1}$	&	$bias_{2}$	&	$bias_{3}$	&	Correct	&	Over	&	Under	&	ME	&	C	&	IC	\\\hline
Mixed Norm	&	DR	&	0.093	&	0.125	&	0.108	&	0.990	&	0.010	&	0.000	&	0.056	&	3.000	&	0.010	\\
	&	MR	&	0.160	&	0.231	&	0.184	&	0.625	&	0.325	&	0.000	&	0.226	&	3.000	&	0.695	\\
	&	QR	&	0.123	&	0.130	&	0.132	&	0.125	&	0.875	&	0.000	&	0.256	&	3.000	&	2.335	\\
$N(0,1)$	&	DR	&	0.188	&	0.354	&	0.290	&	0.795	&	0.010	&	0.195	&	0.810	&	2.755	&	0.015	\\
	&	MR	&	0.404	&	0.529	&	0.509	&	0.275	&	0.590	&	0.135	&	1.608	&	2.850	&	1.845	\\
	&	QR	&	0.226	&	0.248	&	0.238	&	0.004	&	0.955	&	0.005	&	1.377	&	2.995	&	4.760	\\
$t(3)$	&	DR	&	0.150	&	0.156	&	0.155	&	0.975	&	0.025	&	0.000	&	0.111	&	3.000	&	0.030	\\
	&	MR	&	0.152	&	0.165	&	0.165	&	0.630	&	0.365	&	0.005	&	0.183	&	2.990	&	0.725	\\
	&	QR	&	0.122	&	0.119	&	0.117	&	0.150	&	0.850	&	0.000	&	0.146	&	3.000	&	1.820	\\
$Lap(0,1)$	&	DR	&	0.120	&	0.158	&	0.139	&	0.965	&	0.035	&	0.000	&	0.094	&	3.000	&	0.045	\\
	&	MR	&	0.114	&	0.144	&	0.124	&	0.745	&	0.255	&	0.000	&	0.106	&	3.000	&	0.750	\\
	&	QR	&	0.102	&	0.098	&	0.095	&	0.245	&	0.755	&	0.000	&	0.096	&	3.000	&	1.550	\\
$G(2,2)$	&	DR	&	0.042	&	0.045	&	0.046	&	1.000	&	0.000	&	0.000	&	0.009	&	3.000	&	0.000	\\
	&	MR	&	0.061	&	0.059	&	0.058	&	0.880	&	0.120	&	0.000	&	0.021	&	3.000	&	0.325	\\
	&	QR	&	0.060	&	0.059	&	0.059	&	0.730	&	0.270	&	0.000	&	0.023	&	3.000	&	0.355	\\
$Cau(0,1)$	&	DR	&	1.644	&	1.154	&	1.312	&	0.260	&	0.005	&	0.735	&	8.641	&	1.170	&	0.010	\\
	&	MR	&	1.670	&	1.233	&	1.334	&	0.070	&	0.180	&	0.750	&	9.784	&	1.315	&	0.990	\\
	&	QR	&	0.153	&	0.156	&	0.141	&	0.065	&	0.935	&	0.000	&	0.359	&	3.000	&	2.910	\\
$Exp(1)$	&	DR	&	0.044	&	0.043	&	0.038	&	1.000	&	0.000	&	0.000	&	0.009	&	3.000	&	0.000	\\
	&	MR	&	0.090	&	0.082	&	0.083	&	0.745	&	0.255	&	0.000	&	0.050	&	3.000	&	0.685	\\
	&	QR	&	0.082	&	0.078	&	0.076	&	0.550	&	0.450	&	0.000	&	0.042	&	3.000	&	0.640	\\\hline
\end{tabular}
\end{table*}

\subsection{Real Data Analysis}
In the section, we first illustrate the usefulness of the proposed method in the criterion of prediction error via three real data sets:  airmay, aircraft and eyedata.
The data set airmay is used to study the relationship between ozone concentration and three explanatory variables: radiation level, wind speed, and temperature.  It was collected from May 1 to September 30 in 1973. %, with 31 observations.
The aircraft is used to  analyse the relation between cost and four explanatory variables: aspect ratio, lift-to-drag ratio, weight  and thrust, and deals with 23 single-engine
aircraft built over the years 1947-1979. %,  with 23 observations
The eyedata is used to find the genes whose expressions are associated with  the expression level of TRIM32 gene (\cite{Scheetz:2006}). The sample size is 120 and the number of  predictors is 200 (120 rats with 200 gene probes).

 The missing values of airmay are removed. The airmay  and aircraft can be obtained in  R package ``robustbase".  The third data can be found in R package ``flare". Since the sample size is less than  the number of the predictors in  eyedata, we use adaptive LASSO to select variables. In order to eliminate  the effect of scale, we have scaled and centered each of the explanatory variables of airmay and  aircraft such that they have mean 0 and variance 1. We compare the performance of our distribution regression with the mean regression and quantile regression via the following criteria.
\begin{enumerate}[(1)]
  \item We divide the data into the training set and test set. Test set consists of the last $\lceil n/9 \rceil$ numbers of the data set, where $\lceil a \rceil $ denotes the smallest integer not less than $a$.
  \item We use the prediction error, $PE=\lceil n/9 \rceil ^{-1}\sum_{i=n-\lceil n/9 \rceil+ 1}^{n}(Y_{i}-\widehat{Y}_{i})^{2} $ to evaluate their performance, where $\widehat{Y}_{i}$ is the estimated $Y_i$.
\end{enumerate}

Table \ref{Tab5} summarizes the results. The distribution regression performs better than the other two because the PE of distribution regression is smaller than the corresponding mean and quantile regression ones. In fact, the empirical densities of these data sets are either asymmetrical or multimodal. This observation is consistent with our simulation results. Furthermore, for the eyedata, the distribution regression, mean regression and quantile (median) regression select different variables, which are $(x_{153}, x_{180} )$, $(x_{87}, x_{153})$  and $(x_{153}, x_{181})$  respectively.
\begin{table*}[!h]
\centering
\caption{The results of  real data.}
\label{Tab5}
\begin{tabular}{crrr rrrc}
\hline
data	&	$n$	&	$p$	&	DR	&	MR	&	QR	\\\hline
airmay	&	24	&	3	&	3.8747	&	4.8488	&	5.2551	\\
aircraft	&	23	&	4	&	0.8087	&	1.1792	&	0.9416	\\
eyedata	&	120	&	200	&	0.0093	&	0.0120	&	0.0118	\\\hline
\end{tabular}
\end{table*}

We then study Coleman data in diagnosis. This data contains information on 20 schools from the Mid-Atlantic and New England States, drawn from a population studied by \cite{Coleman:1966}. Following \cite{Rousseeuw:1987}, the response and explanatory variables are as follows:
\begin{itemize}
\setlength{\itemsep}{0pt}
\setlength{\parsep}{0pt}
\setlength{\parskip}{0pt}
  \item[]$Y$: verbal mean test score (all sixth graders).
  \item[]$x_{1}$: staff salaries per pupil.
  \item[]$x_{2}$: percent of white-collar fathers.
  \item[]$x_{3}$: socioeconomic status composite deviation: means for family size, family intactness, father's education, mother's education, and home items.
  \item[]$x_{4}$: mean teacher's verbal test score.
  \item[]$x_{5}$: mean mother's educational level, one unit is equal to two school years.
\end{itemize}

%the DR, MR and QR estimates \.
%and p-value of  test result for the error, where where sd.test, cvm.test and sf.test are Anderson-Darling
%test, Cramer-von Mises test  and Shapiro-Francia test for normality. From sd.test and cvm.test, we observe the term error
%of DR  is the normal distribution under 0.05  significance level.
Table \ref{Tab6} shows the estimated coefficients of $x_{1},\dots,x_{4}$ are not much different for all three methods. However, the DR estimate of $x_{5}$ is positive, while the others are negative. We believe the former estimate is more reasonable because the higher the mother's educational level is,  the higher test score should be in a general way. This indicates that DR makes more sense in this particular data. Moreover, Figure \ref{fig1}  shows that the response $Y$ and three regression residuals seems multimodal distributed.

\begin{table*}[!h]
\centering
\caption{The estimated coefficients of  coleman  data} %and  normality test}
\label{Tab6}
\begin{tabular}{crrr rrrc}
\hline
	&	DR	&	MR	&	QR	\\
Intercept	&	0.023	&	19.949	&	29.162	\\
$x_{1}$	&	-2.014	&	-1.793	&	-1.734	\\
$x_{2}$	&	0.012	&	0.044	&	0.059	\\
$x_{3}$	&	0.505	&	0.556	&	0.675	\\
$x_{4}$	&	1.475	&	1.110	&	1.124	\\
$x_{5}$	&	0.240	&	-1.811	&	-3.551	\\\hline
%ad.test &   0.319   &   0.138   &0.000\\
%cvm.test &  0.327   &   0.183   &0.000\\
%sf.test&    0.279   &   0.049   &0.000\\\hline
\end{tabular}
\end{table*}

\begin{figure}[!h]
  \centering
  %\bibliographystyle[width=80mm]{figure1.eps}
  \includegraphics[width=100mm]{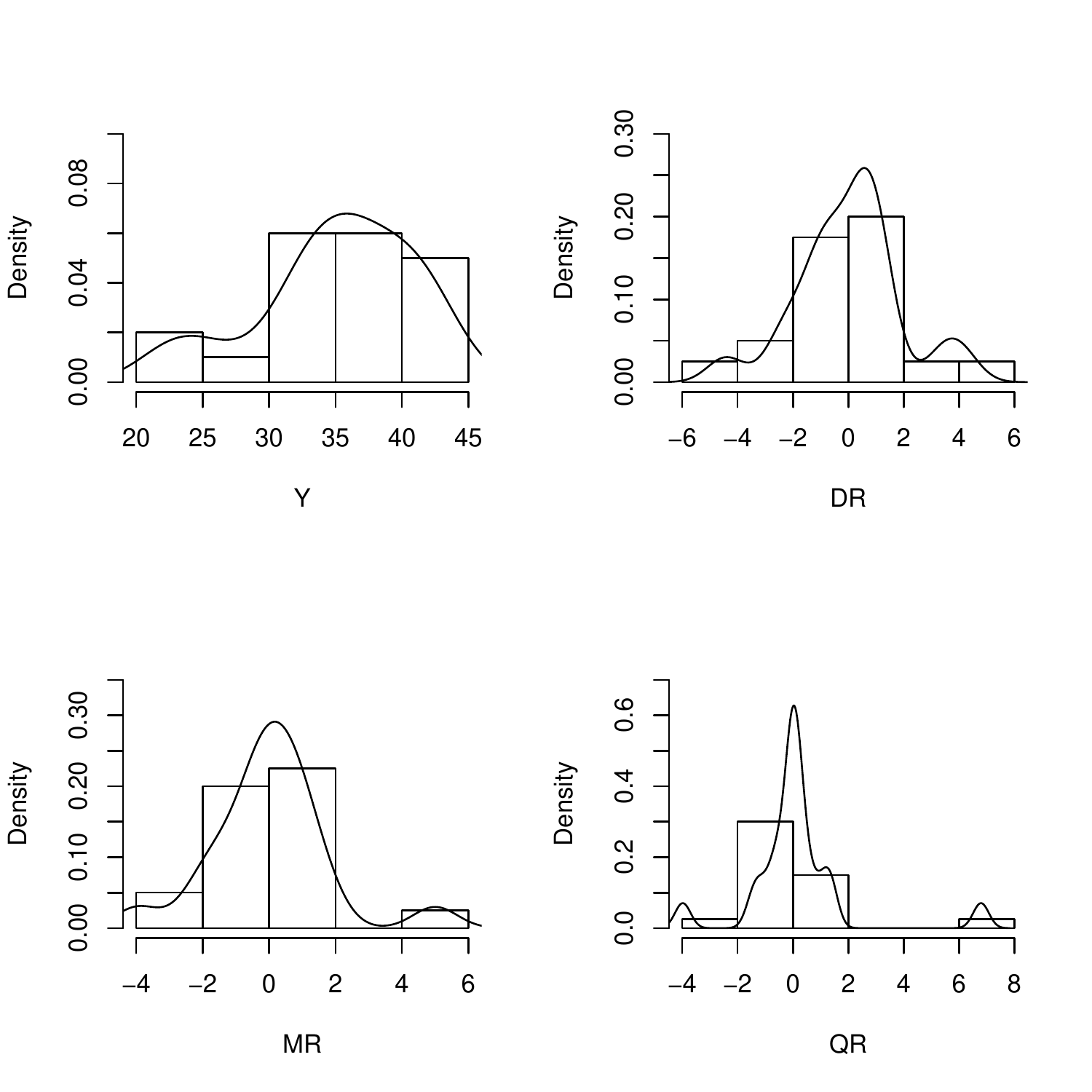}
  \caption{The density estimate of the response and the three regression residuals}\label{fig1}
\end{figure}

%\begin{figure}[!h]
%  \centering
%  % Requires \usepackage{graphicx}
%  %\bibliographystyle[width=80mm]{figure1.eps}
%  \includegraphics[width=100mm]{figure1.eps}
%  \caption{The density estimate of the response and the three regression residuals}\label{fig1}
%\end{figure}

\appendix
\renewcommand{\theequation}{A.\arabic{equation}}
\setcounter{equation}{0}
\section*{Appendix section}\label{app}

To prove Theorem \ref{theorem1}, we need three lemmas as a first step.

\begin{lemma}
\label{lemma1}
Under Conditions \ref{condition4} and \ref{condition5}, if $p$ is fixed,
\begin{equation} \label{A0}
\sup_{\beta\in B(\kappa_0)}|n^{-1}\sum_{i=1}^n \log f(Y_i-x_i^\top\beta)-E_{\beta_0}\log f(Y-x^\top\beta)| \to 0 \ \ a.s.
\end{equation}
as $n\to \infty$; if $p=o(n^{1-\alpha}h/\log n)$ as $n\to \infty$, and $\sup_{\beta\in B(\kappa_0),p\geq 1} E_{\beta_0}\exp\big(\psi$ $|\log(Y-x^\top\beta)|\big)<\infty$ for some positive $\psi$, \eqref{A0} is still correct.
\end{lemma}
\begin{proof}
Let $p$ be fixed. For any $\epsilon>0$, let $\{\beta_1,\cdots,\beta_M\}$ be an $\epsilon/\sqrt p$-net of $B(\kappa_0)$, where $M=([2\kappa_0\sqrt p/\epsilon+1])^p$. In other words,  for any $\beta\in B(\kappa_0)$, there exists one $\beta_j$ such that $\|\beta-\beta_j\|\leq \epsilon/\sqrt p$. Since $M$ is fixed, it follows from the strong law of large numbers that
\begin{equation} \label{A0-1}
\max_{1\leq j\leq M}|n^{-1}\sum_{i=1}^n \log f(Y_i-x_i^\top \beta_j)-E_{\beta_0}\log f(Y-x^\top \beta_j)| \to 0 \ \ a.s.
\end{equation}
as $n\to \infty$. For any $\beta\in B(\kappa_0)$, supposing $\|\beta- \beta_1\|\leq \epsilon/\sqrt p$, we have %Noting that
\begin{eqnarray*}
&&n^{-1}|\sum_{i=1}^n \log f(Y_i-x_i^\top\beta)-\sum_{i=1}^n \log f(Y_i-x_i^\top \beta_1)| \\
&&\leq n^{-1}\sum_{i=1}^n M(U_i)\|x_i/\sqrt p\|\epsilon\to E_{\beta_0}\big(M(U_1)\|x_1/\sqrt p\|\big) \epsilon \ \ a.s.,\\
&&|E_{\beta_0}\log f(Y-x^\top\beta)-E_{\beta_0}\log f(Y-x^\top \beta_1)|\leq E_{\beta_0}\big(M(U_1)\|x_1/\sqrt p\|\big) \epsilon.
\end{eqnarray*}
So we can complete the proof.

Now let $p=o(n^{1-\alpha}h/\log n)$ as $n\to \infty$. Since $\sup_{\beta\in B(\kappa_0)}E_{\beta_0}\exp\big(\psi|\log f(Y-x^\top\beta)|\big)<\infty$, the large deviation principle shows that
for any $\delta>0$, there exists $\epsilon>0$, independent of $n$ such that
$$
\max_{1\leq j\leq M}P\Big(|n^{-1}\sum_{i=1}^n \log f(Y_i-x_i^\top \beta_j)-E_{\beta_0}\log f(Y-x^\top \beta_j)|\geq \delta\Big) \leq e^{-n\epsilon}.
$$
Hence \eqref{A0-1} is still correct. By Condition \ref{condition5} and the Borel-Cantelli lemma,
$$n^{-1}\sum_{i=1}^n M(U_i)\|x_i/\sqrt p\| - E_{\beta_0}\big(M(U_1)\|x_1/\sqrt p\|\big) \to 0 \ a.s.$$
as $n\to \infty$. So the proof for the fixed $p$ case can be adapted to this divergent one. We can complete the proof.
\end{proof}

\begin{lemma}
\label{lemma2}
 Under Conditions \ref{condition1}--\ref{condition4}, if $p=o(n^{1-\alpha}h/\log n)$ as $n\to \infty$, for any $\delta >0$, we have
$$
P\Big{(} \sup_{\beta\in B(\kappa_0)}\Big{|} n^{-1}\sum_{i=1}^n \log f_{nh}(Y_i-x_i^\top\beta)-n^{-1}\sum_{i=1}^n \log f(Y_i-x_i^\top\beta)  \Big{|}  \geq \delta  \Big{)}  \rightarrow 0
$$
as $n\to \infty$.
\end{lemma}
\begin{proof}
We write $n^{-1}\sum_{i=1}^n \log f_{nh}(Y_i-x_i^\top\beta)-n^{-1}\sum_{i=1}^n \log f(Y_i-x_i^\top\beta)=A_1+A_2+A_3$, where
\begin{eqnarray*}
A_1&=&n^{-1}\sum_{i=1}^n \big(\log f_{nh}(Y_i-x_i^\top\beta)- \log f(Y_i-x_i^\top\beta)\big)\\
&& \qquad \qquad\qquad\qquad \times I(\min_{\beta\in B(\kappa_0)}f(Y_i-x_i^\top\beta)\leq n^{-k}),\\
A_2&=&n^{-1}\sum_{i=1}^n \big(\log f_{nh}(Y_i-x_i^\top\beta)- \log f(Y_i-x_i^\top\beta)\big)\\
&& \qquad \qquad\qquad\qquad \times I(n^{-k}<\min_{\beta\in B(\kappa_0)}f(Y_i-x_i^\top\beta)\leq n^{-\alpha}),\\
A_3&=&n^{-1}\sum_{i=1}^n \big(\log f_{nh}(Y_i-x_i^\top\beta)- \log f(Y_i-x_i^\top\beta)\big)\\
&&\qquad \qquad\qquad\qquad \times I(\min_{\beta\in B(\kappa_0)}f(Y_i-x_i^\top\beta)\geq n^{-\alpha}).
\end{eqnarray*}
For $A_1$, we have,
\begin{eqnarray}
P( \sup_{\beta\in B(\kappa_0)}|A_1|>\delta/3) &\leq& P(\min_{1\leq i\leq n}\min_{\beta\in B(\kappa_0)}f(Y_i-x_i^\top\beta)
\leq n^{-k})\nonumber\\
&\leq& nP(\min_{\beta\in B(\kappa_0)}f(Y-x^\top\beta)\leq n^{-k})\to 0 \label{A1}
\end{eqnarray}
by Condition \ref{condition3} as $n\to \infty$.
The fact that
$n^{-1000}\leq f_{nh}(\varepsilon_{i})\leq n^{-1000}+h^{-1}\|K\|_\infty$ implies
\begin{eqnarray}
&&\sup_{\beta\in B(\kappa_0)}|A_2|\leq \Big((1001+k+|\log(\|K\|_\infty)|)\log n \Big) \times \nonumber \\
&&\qquad \qquad\qquad n^{-1}\sum_{i=1}^n I(n^{-k}<\min_{\beta\in B(\kappa_0)}f(Y_i-x_i^\top\beta)\leq n^{-\alpha}). \label{A2-1}
\end{eqnarray}
By Condition \ref{condition3},
$$
EI(n^{-k}<\min_{\beta\in B(\kappa_0)}f(Y_i-x_i^\top\beta)\leq n^{-\alpha})=o(\log^{-1}n)
$$
as $n\to \infty$. Hence it follows from Chebyshev's inequality that
$$
n^{-1}\sum_{i=1}^n I(n^{-k}<\min_{\beta\in B(\kappa_0)}f(Y_i-x_i^\top\beta)\leq n^{-\alpha})=o_p(\log^{-1}n))
$$
as $n\to\infty$. This combined with \eqref{A2-1} gives
\begin{equation} \label{A2}
P(\sup_{\beta\in B(\kappa_0)}|A_2|>\delta/3) \to 0
\end{equation}
as $n\to \infty$. Finally we handle $A_3$. We will prove
\begin{equation} \label{A3}
P\Big{(}   \max_{1\leq i \leq n}\sup_{\beta\in B(\kappa_0)} \Big{|}  \frac{f_{nh}(Y_i-x_i^\top\beta))}{f(Y_i-x_i^\top\beta))}  -1 \Big{|}I(\min_{\beta\in B(\kappa_0)}f(Y_i-x_i^\top\beta)\geq n^{-\alpha}) > \delta/3 \Big{)} \to 0
\end{equation}
as $n\to \infty$, which immediately gives $P(\sup_{\beta\in B(\kappa_0)}|A_3|>\delta/3) \to 0$ as $n\to \infty$.  In fact it suffices to prove
\begin{equation} \label{A3-1}
nP\Big{(}  \sup_{\beta\in B(\kappa_0)} \Big{|}  \frac{f_{nh}(Y_1-x_1^\top\beta))}{f(Y_1-x_1^\top\beta))}  -1 \Big{|}I(\min_{\beta\in B(\kappa_0)}f(Y_1-x_1^\top\beta)\geq n^{-\alpha}) > \delta/3 \Big{)} \to 0
\end{equation}
as $n\to \infty$. In order to prove \eqref{A3-1}, we first prove that
 \begin{equation} \label{A3-2}
n\sum_{j=1}^N P\Big{(} \Big{|}  \frac{f_{nh}(Y_1-x_1^\top\beta_j)}{f(Y_1-x_1^\top\beta_j)}  -1 \Big{|}I(\min_{\beta\in B(\kappa)}f(Y_1-x_1^\top\beta)\geq n^{-\alpha}) > \delta/3 \Big{)} \to 0
\end{equation}
as $n\to \infty$, where $\{\beta_1,\cdots,\beta_N\}$
 is the $n^{-c_0}$-net of $B(\kappa)$, and $N=\big([2\kappa_0 n^{c_0}+1]\big)^p$. Here $c_0$ is a constant which will be specified later. In other words, for each $\beta\in B(\kappa_0)$, there exists one $\beta_j$ ($1\leq j\leq N$) such that $\|\beta-\beta_j\|\leq n^{-c_0}$. Bernstein's inequality and Condition \ref{condition2} imply that
$$P\Big{(} \Big{|}  \frac{f_{nh}(Y_1-x_1^\top \beta_j)}{f(Y_1-x_1^\top \beta_j)}  -1 \Big{|}I(\min_{\beta\in B(\kappa_0)}f(Y_1-x_1^\top\beta)\geq n^{-\alpha}) > \delta/3 \Big{)}\leq \exp(-cn^{1-\alpha}h),
$$
where $c$ is a constant independent of $n$ and $j$. £¨Here we used the condition that $f'(z)$ and $f''(z)$ exist a.e. and have finite $L_\infty$ norms.) This combined with $p=o(n^{1-\alpha}h/\log n)$ directly gives \eqref{A3-2}. Now for any $\beta\in B(\kappa_0)$, we suppose $\|\beta-\beta_1\|\leq n^{-c_0}$. So
\begin{eqnarray*}
|f_{nh}(Y_1-x_1^\top\beta)-f_{nh}(Y_1-x_1^\top \beta_1)|\leq \|K'\|_\infty \|x_1\|h^{-2}n^{-c_0},\\
|f(Y_1-x_1^\top\beta)-f(Y_1-x_1^\top \beta_1)|\leq \|f'\|_\infty \|x_1\|n^{-c_0}.
\end{eqnarray*}
Noting that $E\|x_1/p\|<\infty$, we can select $c_0$ which is large enough and fixed so that $\|x_1\|h^{-2}n^{-c_0}=o_p(n^{-\alpha})$ and $\|x_1\|n^{-c_0}=o_p(n^{-\alpha})$.
The above two inequalities together with \eqref{A3-2} give \eqref{A3-1}.  Combining \eqref{A1}, \eqref{A2}, \eqref{A3} and \eqref{A3-1} completes the proof.
\end{proof}

\begin{lemma}
\label{lemma3}
 Under Conditions \ref{condition1}--\ref{condition8}, for any $\delta >0$, if $p$ is fixed, we have
\begin{eqnarray*}
&&\sup_{\beta\in C(\kappa_0)}|n^{-1}\sum_{i=1}^n \log f(Y_i-x_i^\top\beta)-E_{\beta_0}\log f(Y-x^\top\beta)| \to 0 \ \ a.s., \\
&&P\Big{(} \sup_{\beta\in C(\kappa_0)}\Big{|} n^{-1}\sum_{i=1}^n \log f_{nh}(Y_i-x_i^\top\beta)-n^{-1}\sum_{i=1}^n \log f(Y_i-x_i^\top\beta)  \Big{|}  \geq \delta  \Big{)}  \rightarrow 0
\end{eqnarray*}
as $n\to \infty$. If $p=o(n^{1-\alpha}h/\log n)$ as $n\to \infty$, and $\sup_{\beta\in C(\kappa_0),p\geq 1} E_{\beta_0}\exp\big(\psi$ $|\log(Y-x^\top\beta)|\big)<\infty$ for some positive $\psi$ and $\kappa$, the above first relation is still correct. If $p=o(n^{1-\alpha}h/\log n)$ as $n\to \infty$, the above second relation is still correct. (In this lemma, we don't require analyticity in $C(\kappa_0)$.)
\end{lemma}
\begin{proof}
The proof is similar to those of Lemmas \ref{lemma1} and \ref{lemma2}. The only difference is that in \eqref{A3-2}, which becomes
 \begin{equation} \label{C1}
n\sum_{j=1}^M P\Big{(} \Big{|}  \frac{f_{nh}(Y_1-x_1^\top\beta_j)}{f(Y_1-x_1^\top\beta_j)}  -1 \Big{|}I(\min_{\beta\in C(\kappa)}\|f(Y_1-x_1^\top\beta)\|\geq n^{-\alpha}) > \delta/3 \Big{)} \to 0
\end{equation}
as $n\to \infty$, where $\{\beta_1,\cdots,\beta_M\}$
 is the $n^{-c_0}$-net of $C(\kappa_0)$, and $M=\big([2\kappa_0 n^{c_0}+1]\big)^{2p}$. Here $c_0$ is a constant which will be specified later. In other words, for each $\beta\in C(\kappa_0)$, there exists one $\beta_j$ ($1\leq j\leq M$) such that $\|\beta-\beta_j\|\leq n^{-c_0}$. Berstein's inequality and Condition \ref{condition2} imply that
\begin{eqnarray*}
&&P\Big{(} \Big{|} \Re \big(f_{nh}(Y_1-x_1^\top \beta_j)-f(Y_1-x_1^\top \beta_j)\big) \Big{|}> \\
&&\qquad \frac 16\delta \|f(Y_1-x_1^\top\beta_j)\|
\big{|}
\|f(Y_1-x_1^\top\beta)\|\geq n^{-\alpha}  \Big{)}\leq \exp(-cn^{1-\alpha}h),\\
&&P\Big{(} \Big{|} \Im \big(f_{nh}(Y_1-x_1^\top \beta_j)-f(Y_1-x_1^\top \beta_j)\big) \Big{|}> \\
&&\qquad \frac 16\delta \|f(Y_1-x_1^\top\beta_j)\|
\big{|}
\|f(Y_1-x_1^\top\beta)\| \geq n^{-\alpha} \Big{)}\leq \exp(-cn^{1-\alpha}h)
\end{eqnarray*}
where $c$ is a constant independent of $n$ and $j$. The above implies \eqref{C1}.
\end{proof}

\begin{lemma}
\label{lemma4}
Let $(T_{1n},\dots,T_{pn})$ be a sequence of random vectors tending weakly to $(T_{1},\dots,T_{p})$ and suppose that for each fixed $j$  and $k$, $A_{jkn}$ is a sequence of random variables tending in probability to constants $a_{jk}$ for which the matrix $A=( a_{jk} )$ is nonsingular. Let $B=( b_{jk} )=A^{-1}$. Then if the distribution of $(T_{1},\dots,T_{p})$  has a density with respect to Lebesgue measure over $R_{p}$, the solutions $(V_{1n},\dots,V_{pn})$ of
\begin{equation}\label{eq150}
\sum_{k=1}^{p} A_{jkn}V_{kn} =T_{jn}  ~~j=1,\dots, p
\end{equation}
tend in probability to the solutions $(V_{1},\dots,V_{p})$  of
$$
\sum_{k=1}^{p} a_{jk}V_{k} =T_{j}  ~~j=1,\dots, p
$$
given by

$$
V_{j}=\sum_{k=1}^{p}b_{jk} T_{k}.
$$
\end{lemma}

 The proof of Lemma \ref{lemma4} can be found in \cite{Lehmann:1998}.

\begin{proof}[ Proof of Theorem~\ref{theorem1}]

\begin{align*}
  &\frac{1}{n}\sum_{i=1}^{n} \log f_{nh}(Y_i-x_i^\top \beta) \\
  & =   \frac{1}{n}\sum_{i=1}^{n} \Big{[} \log f_{nh}(Y_i-x_i^\top\beta)-  \log f(Y_i-x_i^\top\beta) +  \log f(Y_i-x_i^\top\beta)  \Big{]}  \\
   &=  \frac{1}{n}\sum_{i=1}^{n} \Big{[} \log f_{nh}(Y_i-x_i^\top\beta)-  \log f(Y_i-x_i^\top\beta)  \Big{]}   + \frac{1}{n}\sum_{i=1}^{n} \log f(Y_i-x_i^\top\beta)\\
   &= D_{1} + D_{2}.
\end{align*}
 By Lemma \ref{lemma2} we obtain that, for any $\kappa\in (0,\kappa_0]$, $\sup_{\beta\in B(\kappa)}|D_1| \rightarrow 0$ in probability as $n\to \infty$.
By Lemma \ref{lemma1}
$$
\sup_{\beta\in B(\kappa)}|D_2-E_{\beta_0}\log f(Y-x^\top\beta)| \to 0 \ \ a.s.
$$
as $n\to \infty$. %Jensen's inequality and Condition \ref{condition4} imply that
% for all $\beta\neq \beta_0\in \omega$,
%\begin{align}
%  E_{\beta_0}\Big{(} \log\frac{f(Y-x^\top\beta)}{f(Y-x^\top\beta_0)} \Big{)} < 0  \label{eq100}
%\end{align}
Combining all the above and \eqref{C24-1}, we can see that
for any sufficiently small  $ \kappa >0 $,
$$
P_{\beta_{0}}\big(\sup_{\beta\in \partial B(\kappa)}l_{nh}(\beta ) < l_{nh}(\beta_{0})\big)\to 1
$$
as $n\to \infty$.
So there exists a local maximum point  with probability tending to 1, that is
$$
\frac{\partial \ell_{nh}(\beta)}{\partial \beta} \Big{|}_{\widehat{\beta}}=0.
$$
Moreover,  we have $P_{\beta_{0}}(\|\widehat{\beta}- \beta_{0}\| < \kappa) \rightarrow 1 $.
 This completes the proof.
\end{proof}

\begin{proof}[ Proof of Theorem~\ref{theorem2}]
Let $\beta=(b_1,\cdots,b_p)^\top$, $\beta_0=(b_{10},\cdots,b_{p0})^\top$, where $b_1,\cdots,b_p$ are complex.
Expand $\partial\ell_{nh} (\beta)/\partial b_{j}=\ell'_{nh,j}(\beta)$ about $\beta_{0}$ to obtain
\begin{align*}
  \ell'_{nh,j}(\beta) & =   \ell'_{nh,j}(\beta_{0}) + \sum_{k=1}^{p} (b_{k}-b_{k0})\ell''_{nh,jk}(\beta_{0}) \\
   & + \frac{1}{2}\sum_{k=1}^{p} \sum_{l=1}^{p} (b_{k}-b_{k0})(b_{l}-b_{l0}) \ell'''_{nh,jkl}(\beta^{\ast}),
\end{align*}
where $\ell''_{nh,jk}$  and $\ell'''_{nh,jkl}$ denote the indicated second and third derivatives of $\ell_{nh}$, and $\beta^{\ast}$ is a point on the line segment connecting $\beta$ and $\beta_{0}$. In the expansion, replace $\beta$ by $\widehat{\beta}=(\widehat b_1,\cdots,\widehat b_p)$, which by  Theorem \ref{theorem1}
can be assumed to exist with probability tending to 1 and to be consistent. Hence, $\ell'_{nh,j}(\widehat{\beta})=0 $, and
\begin{align*}
  \sqrt{n} & \sum _{k=1}^{p} (\widehat{b}_{k}-b_{k0})\Big{[}   \frac{1}{n}\ell''_{nh,jk}(\beta_{0}) + \frac{1}{2n}\sum_{l=1}^{p} (\widehat b_{l}-b_{l0}) \ell'''_{nh,jk}(\beta^{\ast}) \Big{]}  \\
  &=- \frac{1}{\sqrt{n}}\ell'_{nh,j}(\beta_{0}).
\end{align*}
This is just (\ref{eq150}) in Lemma \ref{lemma4} with
\begin{eqnarray*}
  V_{kn} &=& \sqrt{n}(\widehat{b}_{k}-b_{k0}) \\
  A_{jkn} &=&  \frac{1}{n}\ell''_{nh,jk}(\beta_{0}) + \frac{1}{2n}\sum_{l=1}^{p} (\widehat{b}_{l}-b_{l0}) \ell'''_{nh,jk}(\beta^{\ast})  \\
  T_{jn} &=& - \frac{1}{\sqrt{n}}\ell'_{nh,j}(\beta_{0}) =-\sqrt{n}\Big{[}  \frac{1}{n}  \sum_{i=1}^{n} \frac{\partial}{\partial b_{j}}\log f_{nh}(Y_i-x_i^\top\beta)\Big{]}_{\beta=\beta_{0}}
\end{eqnarray*}
By Lemma \ref{lemma3}, $n^{-1}l_{nh}$ converges to $E_{\beta_0}\log f(Y-x^\top \beta)$ in probability. Since $n^{-1}l_{nh}$ and $E_{\beta_0}\log f(Y-x^\top \beta)$ are analytic functions of each coordinate of $\beta$ in $C(\kappa_0)$, it follows from convergence of analytic functions and Theorem \ref{theorem1} that
$$
A_{jkn} \rightarrow a_{jk} = \frac{\partial^2}{\partial b_j\partial b_k}E_{\beta_0}\log f(Y-x^\top \beta)= -\bm{I}_{jk}(\beta_{0})
$$
in probability as $n\to \infty$. It remains to show the asymptotic normality of $(T_{1n},\dots, T_{pn})$.  We will only consider $T_{1n}$. The general case can be handled by linear combinations of $T_{1n},\dots, T_{pn}$.

Jensen's inequality shows that $E_{\beta_0} \log f(Y-x^\top \beta_0)<0$. By Lemma \ref{lemma3}, there exists $\delta>0$ such that
$$
\inf_{\beta\in C(\kappa_0)}|n^{-1}\sum_{i=1}^n\log f(Y_i-x_i^\top\beta)|\geq \delta \ \ a.s.
$$
 as $n\to \infty$. Hence by Lemma \ref{lemma3} again, we have
$$
\sum_{i=1}^n \log f_{nh}(Y_i-x_i^\top \beta)=\sum_{i=1}^n \log f(Y_i-x_i^\top \beta)\big(1+u_n(\beta)+{\rm i} v_n(\beta)\big)
$$
 when $\beta\in C(\kappa_0)$, $\rm i^2=-1$, and
 \begin{equation} \label{CLT0}
 u_n(\beta)+{\rm i}v_n(\beta)=o_p(1)
 \end{equation}
  is uniform in $\beta\in C(\kappa_0)$ as $n\to \infty$. Here both $u_n(\beta)$ and $v_n(\beta)$ are real functions of $\beta$.
 Note that Cauchy's residue theorem gives
$$
T_{1n}=\frac{1}{\sqrt n}\frac{1}{2\pi\rm i}\oint_{\mathcal C_1}\frac{\sum_{i=1}^n \log f_{nh}(Y_i-x_i^\top \beta_1)}{(b_1-b_{10})^2}db_1,
$$
where  $\beta_1=(b_1,b_{20},\cdots,b_{p0})$, and $\mathcal C_1$ is the counterclockwise oriented contour $\{b_1\in {\bf C}:\|b_1-b_{10}\| = r\}$ with $r<\kappa_0$. (Here $r$ may depend on the sample $\{(x_i,Y_i):i=1,2,\cdots,n\}$. However we only require that $r$ be positive so that we can apply Cauchy's residue theorem.)  So
\begin{align}
 T_{1n}-S_{1n}
=\frac{1}{\sqrt n}\frac{1}{2\pi\rm i}\oint_{\mathcal C_1}\frac{\sum_{i=1}^n \log f(Y_i-x_i^\top \beta_1)\big(u_n(b_1)+{\rm i} v_n(b_1)\big)}{(b_1-b_{10})^2}db_1. \label{CLT1}
\end{align}
where, $u_n(b_1)=u_n(\beta_1), v_n(b_1)=v_n(\beta_1)$, and
$$
S_{1n}=\frac{1}{\sqrt n}\frac{1}{2\pi\rm i}\oint_{\mathcal C_1}\frac{\sum_{i=1}^n \log f(Y_i-x_i^\top \beta_1)}{(b_1-b_{10})^2}db_1.
$$
Denote the real part and imaginary part of $n^{-1/2}\sum_{i=1}^n \log f(Y_i-x_i^\top \beta_1)$ by $R_n(b_1)$ and $I_n(b_1)$ respectively. Since the left-hand side of \eqref{CLT1} is real, we consider the real part of the right-hand side of \eqref{CLT1}, which is equal to,
\begin{equation}\label{CLT2}
\frac 1{2\pi}\int_0^{2\pi} \frac{\big(R_n\cos\theta+I_n\sin\theta\big)u_n+\big(R_n\sin\theta-I_n\cos\theta\big)v_n}{r} d\theta,
\end{equation}
where $R_n, I_n, u_n, v_n$ are actually $R_n(b_{10}+re^{{\rm i}\theta})$, $I_n(b_{10}+re^{{\rm i}\theta})$, $u_n(b_{10}+re^{{\rm i}\theta})$, $v_n(b_{10}+re^{{\rm i}\theta})$ respectively. By the mean value theorem, we have
\begin{equation}\label{CLT3}
\eqref{CLT2}= \frac{\big(R_{n0}\cos\theta_0+I_{n0}\sin\theta_0\big)u_{n0}+\big(R_{n0}\sin\theta_0-I_{n0}\cos\theta_0\big)v_{n0}}{r}
\end{equation}
where $\theta_0\in[0,2\pi)$,  $R_{n0}=R_n(b_{10}+re^{{\rm i}\theta_0}), I_{n0}=I_n(b_{10}+re^{{\rm i}\theta_0}), u_{n0}=u_n(b_{10}+re^{{\rm i}\theta_0}), v_{n0}=v_n(b_{10}+re^{{\rm i}\theta_0})$.
On the other hand,
\begin{equation}\label{CLT4}
S_{1n}=-n^{-1/2}\sum_{i=1}^n \frac{\partial}{\partial b_1} \log f(Y_i-x_i^\top \beta_0)=-\frac{\partial}{\partial b_1}\big(R_n(b_{10})+{\rm i} I_n(b_{10})\big).
\end{equation}
By \eqref{CLT2}, we also have
$$
|\eqref{CLT1}| \leq 2(|R_{n0}/r|+|I_{n0}/r|)(|u_{n0}|+|v_{n0}|).
$$
Letting $r\to 0$ in the above inequality and noticing \eqref{CLT0}, \eqref{CLT4} and the central limit theorem of $S_{1n}$ (Condition \ref{condition7}), we have
$$
\eqref{CLT1}=o_p(1)
$$
as $n\to \infty$. Again using the central limit theorem of $S_{1n}$, we can complete the proof.
\end{proof}

\begin{proof}[Proof of Theorem~\ref{theorem3}]
Let $\beta=\beta_{0} + \frac{\alpha}{\sqrt{n}}$, where $\alpha=(\alpha_{1}^\top, \alpha_{2}^\top)^\top$, $\alpha_{1}=(a_{1},\dots,a_{s})^\top$ and $\alpha_{2}=(a_{s+1},\dots,a_{p})^\top$. Define $Z_{n}(\alpha)$ by
\begin{equation}\label{eq20}
  Z_{n}(\alpha)= Q\Big{(}  \beta_{0} + \frac{\alpha}{\sqrt{n}}  \Big{)} - Q(  \beta_{0}).
\end{equation}
Note that
\begin{equation}\label{eq21}
   Z_{n}(\alpha) = -\ell_{nh}\Big{(} \beta_{0}+ \frac{\alpha}{\sqrt{n}} \Big{)}  + \ell_{nh}( \beta_{0}) + n\lambda_n \sum_{j=1}^{p}\hat{w}_{j}\Big{[}\Big{|}    b_{j0} + \frac{a_{j}}{\sqrt{n}}\Big{|}  -|b_{j0}| \Big{]}.
\end{equation}
Using the Taylor expansion of $\ell_{nh}( \beta_{0}+ \alpha/\sqrt{n})$ at $\beta_{0}$, we have
\begin{align*}
&-\ell_{nh}\Big{(} \beta_{0}+ \frac{\alpha}{\sqrt{n}} \Big{)}  + \ell_{nh}( \beta_{0})\\
& = -\frac{1}{\sqrt{n}} \alpha^\top\nabla \ell_{nh}(\beta_{0}) - \frac{1}{2n}\alpha^\top\nabla^{2} \ell_{nh}(\beta^{\ast}) \alpha
\end{align*}
where $\beta^{\ast}$ lies between $\beta_{0}$ and $\beta_{0}+ \alpha/\sqrt{n}$.

\noindent
1. {\bf Fixed $p$}.

From the proof of Theorem \ref{theorem2}, we have
\begin{align*}
   &  - \frac{1}{\sqrt{n}}\nabla \ell_{nh}(\beta_{0}) \xrightarrow{d} N(0, \bm{I}(\beta_{0})), \\
& \frac{1}{n}\nabla^{2} \ell_{nh}(\beta^{\ast}) \xrightarrow{p} -\bm{I}(\beta_{0})
\end{align*}
as $n\to \infty$.
Hence
$$
-\ell_{nh}\Big{(} \beta_{0}+ \frac{\alpha}{\sqrt{n}} \Big{)}  + \ell_{nh}( \beta_{0})  \xrightarrow{d} \alpha^\top W + \frac{1}{2}\alpha^\top \bm{I}(\beta_{0}) \alpha,
$$
where $W$ follows $N(0, \bm{I}(\beta_{0}))$.

Now, we consider the limiting behavior of the third term in (\ref{eq21}). If $b_{j0}\neq 0~~(j=1,\dots,s)$, then $\hat{w}_{j} \rightarrow |b_{j0}|^{-\gamma}$, and $\sqrt{n} (|b_{j0} +a_{j}/\sqrt{n}|  -|b_{j0}| ) \rightarrow a_{j} \text{sign}(b_{j0})$. So we have
$$n\lambda_{n} \hat{w}_{j} (|   b_{j0} + a_{j}/\sqrt{n}|  -|b_{j0}| )=O(\lambda_{n}n^{1/2}).$$
If $b_{j0} = 0$ and $a_j\not=0$ for $j=s+1,\dots,p$,
$$n\lambda_{n} \hat{w}_{j} (|   b_{j0} + \alpha_{j}/\sqrt{n}|  -|b_{j0}| )= \lambda_{n}n^{1/2}(|\hat{b}_{j}|)^{-\gamma}|a_{j}|\rightarrow \infty.$$
 Note that $\hat{b}_{j}= o_{p}(1)$.  Hence we get
$$ Z_{n}(\alpha) \xrightarrow{d} Z(\alpha)=\left\{
\begin{array}{cl}
\alpha_{1}^\top W_{1} + \frac{1}{2}\alpha_{1}^\top \bm{I}_{11}(\beta_{0}) \alpha_{1}      &   \text{if} ~~a_{j}=0 ~~\text{for}~~ j=s+1, \dots, p, \\
\infty     & \text{otherwise},
\end{array} \right. $$
where  $W_{1}$ follows $N(0, \bm{I}_{11}(\beta_{0}))$. Since $Z_n(\alpha)$ is convex, and has a unique minimizer $\sqrt{n}(\widehat{\beta}- \beta_{0})$, it follows from the epi-convergence results of Geyer \cite{Geyer:1994} that
$$
\hat{\alpha}_{10} \xrightarrow{d}  \bm{I}_{11}(\beta_{0})^{-1}W_{1}. %\text{and} \hat{\bm{u}}_{1} \xrightarrow{d} \bm{0}
$$
Hence we proved the asymptotic normality part.

\noindent
2. {\bf Divergent $p$}.

Note that
$$\frac{1}{\sqrt{n}} \alpha^\top\nabla \ell_{nh}(\beta_{0})=\frac{1}{\sqrt{n}} \sum_{i=1}^n\sum_{j=1}^p a_j \frac{\partial}{\partial b_j}
\log f_{nh}(Y_i-x_i^\top \beta_0)$$ and
$$
Var \sum_{j=1}^p a_j \frac{\partial}{\partial b_j}
\log f(Y_i-x_i^\top \beta_0)=\alpha^\top \bm{I}(\beta_{0})\alpha<\infty.
$$
In the proof of Theorem \ref{theorem2}, after replacing $\frac{\partial}{\partial b_1}
\log f_{nh}(Y_i-x_i^\top \beta_0)$ by $$\sum_{j=1}^p a_j \frac{\partial}{\partial b_j}
\log f_{nh}(Y_i-x_i^\top \beta_0)$$ and following the same proof, we can obtain that
\begin{eqnarray*}
\frac{1}{\sqrt{n}} \alpha^\top\nabla \ell_{nh}(\beta_{0})&=&\frac{1}{\sqrt{n}} \sum_{i=1}^n\sum_{j=1}^p a_j \frac{\partial}{\partial b_j}
\log f(Y_i-x_i^\top \beta_0)+\sqrt p o_p(1)\\
&=&O_p(1)+\sqrt p o_p(1)
\end{eqnarray*}
as $n\to \infty$.
Lemmas \ref{lemma1}, \ref{lemma2} and Cauchy's residue theorem imply that
$$
n^{-1}\sum_{i=1}^n\frac{\partial^2}{\partial b_j\partial b_k}
\log f_{nh}(Y_i-x_i^\top \beta^{\ast}) \to E_{\beta_0} \Big[\frac{\partial^2}{\partial b_j\partial b_k}
\log f(Y_i-x_i^\top \beta_0)\Big]
$$
in probability as $n\to \infty$ uniformly in $j,k$. Hence
$$
\frac{1}{2n}\alpha^\top\nabla^{2} \ell_{nh}(\beta^{\ast}) \alpha=O(1)+pO_p(1)
$$
Now following the proof of fixed $p$ case and noting that $p=O(\lambda_n n^{1/2})$, we can complete the asymptotic part for divergent $p$.

Next we show the consistency.
Let $\eta_{n}=Cn^{-1/2}$, where $C$ is a constant. It is sufficient to show that with probability tending to 1 as $n\rightarrow \infty $, for any $\beta=(\beta_{1}^\top,b_{s+1},\cdots,b_p)^\top$ satisfying $\beta_{1}- \beta_{10} =O_{p}(n^{-1/2})$, we have , for $j=s+1,\dots, p$,
\begin{eqnarray}
  \frac{\partial Q(\beta)}{\partial b_{j}}  & < & 0 ~~\text{for} ~~0<b_{j} < \eta_{n} , \label{eq23}\\
   \frac{\partial Q(\beta)}{\partial b_{j}}  & > & 0 ~~ \text{for} ~~- \eta_{n}< b_{j} < 0. \label{eq24}
\end{eqnarray}
To show (\ref{eq23}), by Taylor expansion,
\begin{align*}
   \frac{\partial Q(\beta)}{\partial b_{j}} & = - \frac{\partial \ell_{nh}(\beta)}{\partial b_{j}} + n\lambda_{n}\hat{w}_{j}\text{sign}(b_{j}) \\
  & = - \frac{ \partial \ell_{nh}(\beta_{0})}{\partial b_{j}} - \sum_{l=1}^{p} \frac{\partial^{2} \ell_{nh}(\beta^{\ast})}{\partial b_{j}\partial b_{l}}(b_{l}-b_{l0})+ n\lambda_{n}\hat{w}_{j}\text{sign}(b_{j})
%  &~~~ - \sum_{l=1}^{p} \sum_{k=1}^{p}  \frac{\partial^{3} \ell_{nh}(\beta^{\ast})}{\partial b_{j}\partial b_{l} \partial b_{k} }(b_{l}-b_{l0})(b_{k}-b_{k0})
\end{align*}
 where $\beta^{\ast}$ lies between $\beta$ and $\beta_{0}$.

 From the proof of Theorem \ref{theorem2}, we have
 \begin{equation*}\label{eq25}
   \frac{1}{n}\frac{\ell_{nh}(\beta_{0})}{\partial b_{j}} = O_{p}(n^{-1/2})
 \end{equation*}
and
 \begin{equation*}\label{eq26}
   \frac{1}{n}  \frac{\partial^{2} \ell_{nh}(\beta^{\ast})}{\partial b_{j}\partial b_{l}} \xrightarrow{p} -\bm{I}_{jl}(\beta_{0})
 \end{equation*}
 as $n\to \infty$.

 By the assumption that $\beta_{1}$ satisfies $\beta_{1}- \beta_{10} =O_{p}(n^{-1/2})$, we have uniformly in $j$
 $$
  \frac{\partial Q(\beta)}{\partial b_{j}}= n[\lambda_{n} \hat{w}_{j}\text{sign}(b_{j}) + pO_{p}(n^{-1/2}) ]
 $$
with $\hat{w}_{j} \xrightarrow{p} |b_{j0}|^{-\gamma}$. So the sign of the derivative is completely determined by that of $b_{j}$, noting $p=O(\lambda_n n^{1/2})$ Hence, (\ref{eq23}) and (\ref{eq24}) follow. This completes the proof.
\end{proof}

\end{document}